\documentclass[12pt,journal,onecolumn]{IEEEtran}
\usepackage[letterpaper, left=54pt, right=54pt, top = 54pt, bottom=54pt]{geometry}
\usepackage{hyperref}
\usepackage{amsthm}
\usepackage{epsfig}
\usepackage{xfrac}
\usepackage{amsmath}
\usepackage{cases}
\usepackage{amssymb,mathrsfs,dsfont}
\usepackage{enumerate}
\usepackage{enumitem}
\usepackage{color}
\usepackage{soul}
\usepackage{textcomp}
\usepackage{mathtools}
\usepackage{scrextend}
\usepackage{stackrel}
\usepackage{dsfont}
\usepackage{xparse}
\usepackage{tikz,pgfplots}
\usepackage{kbordermatrix}
\usepackage[utf8]{inputenc}
\usepackage{xfrac}
\newtheorem{lemma}{Lemma}

\newcommand{\mybf}[1]{{\bf #1}}

\def\eps{\varepsilon}

\def\tn{\textnormal}

\newcommand{\dfn}{\stackrel{\tn{def}}{=}}

\def\p2p{point-to-point}

\IEEEoverridecommandlockouts
\mathtoolsset{showonlyrefs=true}
\usepackage[noadjust]{cite}

\def\protLength{n}

\def\blockLength{m}
\def\nblocks{p}
\def\Atrans{A}
\def\Btrans{B}

\def\BSCprob{\varepsilon}
\def\CBSC{C_{\mathrm{BSC}}}
\def\Rb{R_{{b}}}

\begin{document}
\title{
	Shannon Capacity is Achievable for Binary Interactive First-Order Markovian Protocols}
\author{Assaf Ben-Yishai, Ofer Shayevitz and Young-Han Kim 
	\thanks{A. Ben-Yishai and O. Shayevitz are with the Department of EE--Systems, Tel Aviv University, Tel Aviv, Israel. 
		Y.-H.~Kim is with the Department of Electrical and Computer Engineering, University of California, San Diego, La Jolla, CA 92093 USA.
		Emails: \{assafbster@gmail.com, ofersha@eng.tau.ac.il, yhk@ucsd.edu\}. The work of A. Ben-Yishai was partially supported by an ISF grant no. 1367/14. The work of O. Shayevitz was supported by an ERC grant no. 639573, a CIG grant no. 631983, and an ISF grant no. 1367/14.}}

\maketitle

\begin{abstract}
We address the problem of simulating an arbitrary binary interactive first-order Markovian protocol over a pair of binary symmetric channels with crossover probability $\varepsilon$. We are interested in the achievable rates of reliable simulation, i.e., in characterizing the smallest possible blowup in communications such that a vanishing error probability (in the protocol length) can be attained. Whereas for general interactive protocols the output of each party may depend on all previous outputs of its counterpart, in a (first-order) Markovian protocol this dependence is limited to the last observed output only. 
Previous works in the field discuss broader families of protocols but assess the achievable rates only at the limit where $\varepsilon\to0$.

In this paper, we prove that the one-way Shannon capacity, $1-h(\varepsilon)$, can be achieved for any binary first-order Markovian protocol.
This surprising result, is to the best of our knowledge, the first example in which non-trivial interactive protocol can be simulated in the Shannon capacity. We give two capacity achieving coding schemes, which both divide the protocol into independent blocks, and implement vertical block coding. The first scheme is based on a random separation into blocks with variable lengths. The second scheme is based on a deterministic separation into blocks, and efficiently predicting their last transmission. The prediction can be regarded as a binary \textit{pointer jumping game}, for which we show that the final step can be calculated with $O(\log m)$ bits, where $m$ is the number of rounds in the game.
We conclude the work by discussing possible extensions of the result to higher order models.
\end{abstract}

\section{Introduction} 
Suppose Alice and Bob wish to communicate using some interactive communication protocol, where at time point $i$ Alice sends the bit $A_i$ and Bob, after having observed Alice's transmission, replies with the bit $B_i$. The transcript associated with this protocol is therefore 
\begin{align}
A_1,B_1,A_2,B_2,\cdots,A_n,B_n
\end{align}
where 
\begin{align}\label{eq:FullInteraction}
A_i = f_i\left(\mybf{B}^{i-1}\right);\quad B_i =g_i\left(\mybf{A}^{i}\right). 
\end{align}
The protocol consists of the sets of Alices transmission functions $\mybf{f}^n\dfn\{f_1(\cdot),\ldots,f_n(\cdot) \}$ and Bob's transmission functions $\mybf{g}^n\dfn\{g_1(\cdot),\ldots,g_n(\cdot) \}$.
 
Now, suppose that Alice and Bob are connected by a pair of Binary symmetric channels (BSC's) with a crossover probability $\eps$. Namely, the channel from Alice to Bob is
\begin{align}
Y_j=X_j+Z_j
\end{align}
and the channel from Bob to Alice is 
\begin{align}
V_j=U_j+W_j.
\end{align}
The symbol ``$+$" denotes addition over $\mathbb{GF}(2)$ and $Z_j$, $W_j$ 
are mutually independent Bernoulli i.i.d. sequence with $\Pr(Z_j=1)=\Pr(W_i=1)=\eps$.
Alice and Bob would like to devise a coding scheme that will enable a reliable simulation of the original transcript over the noisy BSC's. Reliable simulation in this context means that for any protocol, the probability of either Alice or Bob making an error in recovering the original transcript goes to zero with the transcript length.
The coding scheme would comprise Alice's and Bob's transmission policies. Namely, the coding scheme will determine Alice's and Bob's transmissions ($X_j$ and $U_j$ respectively)  which depend on their transmission
functions ($\mybf{f}^n$ and $\mybf{g}^n$ respectively) and their corresponding sets of previously received inputs. Since the transmissions can depend on the received inputs, the coding scheme itself might also be interactive. It is important to note that while in the original transcript the order of speakers was alternating on a symbol basis (Alice, Bob, Alice, Bob etc.), this is not necessarily the case with the order of transmissions over BSC. Specifically, Alice or Bob can use the channel for several consecutive channel uses while the other party is silent. We count the total number of channel uses on either channel and denote it by $\tilde{n}$. 

The {rate} of any communication scheme that attempts to simulate the original transcript is therefore
\begin{align}\label{eq:ratedef}
R = \frac{2n}{\tilde{n}}
\end{align}
where $2n$ is the length of the original transcript. As usual, one is interested in characterizing the \textit{capacity}, namely the maximal rate for which reliable simulation is possible.   

The problem described above was originally introduced and studied by Schulman \cite{schulman1996coding}. In this seminal work, he showed that reliable simulation with a positive rate (i.e., a positive capacity) can be achieved for any $\eps\neq 1/2$. Kol and Raz \cite{kol2013interactive} further studied the problem in the limit of $\eps\to 0$ and introduced a scheme achieving a rate of $1-O(\sqrt{h(\eps)})$ (where $h(\cdot)$ denotes the binary entropy function). They also showed that for a larger class of protocols with non-alternating rounds the rate is upper bounded by $1-\Omega(\sqrt{h(\eps)})$. These results demonstrate a separation between one-way and interactive communications, as the one-way capacity is given by $1-h(\eps)$. In \cite{haeupler2014interactive}, Haeupler 
examined a more flexible channel model than ours, in which at every time slot Alice and Bob can independently decide if they want to use the channel as a transmitter or as a receiver. This flexibility can potentially lead to collisions, but was shown to eventually increase the achievable rate to $1-O(\sqrt{\eps})$. Haeupler also conjectured that this rate is order-wise tight under adaptive transmission order, i.e., that the rate of any such reliable scheme is upper bounded by $1-\Omega(\sqrt{\eps})$. We note that the general problem of exactly determining the capacity for any fixed $\eps$ in the interactive setup is still wide open. 

In order to better understand the gap between the one-way and interactive setups for $\eps\to 0$, Haeupler and Velingker \cite{haeupler2017bridging} considered a more restrictive family of protocols that are ``less interactive'', where Alice and Bob have some limited average lookahead, i.e., can often speak for a while, without requiring further input from their counterpart (hence, can use short error correcting codes). They showed (also for adversarial noise) that when this average lookahead is $\textrm{poly}(1/\eps)$ then the capacity is $1-O(h(\eps))$, i.e., is order-wise the same as the one-way capacity. 

In this study, rather than restricting the ``interactiveness'' of the protocol as above, we restrict the \textit{memory} of the protocol. Specifically, we consider \textit{Markovian protocols} for which the lookahead can be as short as one (i.e. highly interactive), but where Alice and Bob need only recall the last bit they have received. 
For these Markovian Protocols, we show that reliable simulation is possible at any rate smaller than Shannon capacity. This surprising result is, to the best of our knowledge, the first example in which Shannon capacity is achievable in a non-trivial interactive setup. This result closes the question proposed in our previous study \cite{MarkovInteractionAllerton2017} in which only achievable simulation rates where given. 

\section{Markovian Protocols}
A (first-order) Markovian protocol is a protocol in which each party needs to know only the last transmission of its counterpart in order to decide what to send next, and not the entire set of past transmissions. Namely, 
\begin{align}
A_i =f_i(B_{i-1});\quad B_i =g_i(A_{i}).
\end{align}
where now, in contrast to \eqref{eq:FullInteraction}, the transmission functions  $f_i(\cdot)$, $g_i(\cdot)$ depend only on what was last received ($B_{i-1}$ and $A_{i}$ respectively). 

The probability of error attained by a scheme is defined to be the maximal probability that either Alice or Bob fail to exactly simulate the original transcript, where the maximum is taken over all possible Markovian protocols. A sequence of schemes with rate at least $R$ and error probability approaching zero is said to achieve the rate $R$. The capacity for Markovian protocols over BSC's is the supremum over all such achievable rates, and is denoted by $C_{\textrm{Markov}}(\eps)$. Note that $C_{\textrm{Markov}}(\eps)$ cannot exceed the one-way Shannon capacity of the BSC, i.e.,  
\begin{align}\label{eq:shannon_capacity}
C_{\textrm{Markov}}(\eps) \leq 1-h(\eps), 
\end{align}
as this is the maximal achievable rate for the special case of non-interactive protocols. 
In the sequel we prove that $C_{\textrm{Markov}}(\eps) = 1-h(\eps)$ by giving an explicit capacity achieving coding scheme.

\section{Building blocks for the capacity achieving schemes}
We start by noting that the Markovian transmission functions ${f}_i(\cdot),{g}_i(\cdot)$, are binary functions that map a single input bit to a single output bit. There are only four such functions, $\mu_1$, $\mu_2$, $\mu_3$, $\mu_4$ which are elaborated in Table~\ref{table:markfunctions}.
\begin{table}
\begin{align}
\begin{tabular}{|c||*{4}{c|}}
\hline
&$\mu_1$: & $\mu_2$: & $\mu_3$: & $\mu_4$: \\
$Y$ & $X=Y+ 0$ & $X=Y+ 1$ & $X=0$ & $X=1$ \\
\hline\hline
$0$ & $0$ &  $1$ & $0$ & $1$\\
\hline
$1$ & $1$ &  $0$ & $0$& $1$\\
\hline
\end{tabular}
\end{align}
\caption{Binary first-order Markovian transmission functions\label{table:markfunctions}}
\end{table}
We observe that $\mu_1$ and $\mu_2$ are \textit{additive}, i.e. $X=Y+ \xi$ and $\xi$ is either $0$ or $1$. $\mu_3$ and $\mu_4$ are constant functions, namely, the output is $0$ or $1$ regardless the input. In the sequel we refer to the  $\mu_3$ and $\mu_4$ as ``stuck functions". 
The capacity achieving scheme is based on the properties of these functions and two simple principles termed \textit{non-interactive simulation}, and \textit{block-wise interactive coding} defined herein.

\subsection{Non-Interactive Simulation} 
We observe that if Bob knows all Alice's transmission functions, he can simulate the entire transcript off-line (i.e. without using the channel). Then, he can send his part of the transcript to Alice, which can in-turn simulate her part. 

Since there are only four transmission functions, Alice needs at most two bits for the representation of each function. In addition, if Bob has some side information that reduces the number of functions used by Alice, she can use less than two bits (e.g. if Bob knows that Alice's functions are additive, she can use only a single bit).

A simple non-interactive (and not capacity achieving) scheme can be built as follows: Alice sends Bob all here functions, and then Bob simulates the transcript off-line and feed his part of the protocol to Alice. Alice, needs to send $2n$ bits and Bob needs $n$ bits. Both parties can use capacity achieving block codes for their transmissions so a rate of $2/3$ the Shannon capacity can be trivially achieved. 

\subsection{Block-Wise Interactive Coding}
Let us start with the following simple example. Suppose the transcript is
\begin{align}
A_1, B_1, A_2, B_2, A_3, B_3, A_4, B_4.
\end{align}
The length of the transcript is eight, and both Alice and Bob transmit four times.
The order of speakers is Alice, Bob, Alice, Bob etc. By the Markovity assumption, $B_1$ is a Markovian function of $A_1$, $A_2$ is a Markovian function of $B_1$ and so on.

Now, assume that Alice can explicitly calculate the value of $A_3$ without having to simulate the protocol. For example, assume that $A_3$ is a stuck position (as explained in Section~\ref{sec:randomgame}) 
of that Alice knows $B_2$ before simulating the protocol, and can therefore calculate $A_3$ (as explained in Section~\ref{sec:regularscheme}). Since by the Markovity assumption, $B_3$ is a function of $A_3$ alone (and not of previous transmissions), it is possible to simulate the protocol in the following order:
\begin{align}
\underbrace{A_1, A_3}_{\text{Alice}}, \underbrace{B_1, B_3}_{\text{Bob}},\underbrace{A_2, A_4}_{\text{Alice}},\underbrace{ B_2,B_4}_{\text{Bob}}
\end{align}
where the identity of the speaker is written under the brace. Namely, Alice sends two consecutive transmissions, then Bob sends two consecutive transmissions etc.

This notion is generalized in Table~\ref{table:blocktrans}. We assume that Alice knows her transmission at the beginning of every block ($A_1$, $A_{m+1}$, etc.). Then the protocol can simulated in \textit{vertical blocks} according to the columns of the table. For example, the block contains $A_1, A_{m+1},...,A_{n-m+1}$ (colored blue in the table) sent from Alice to Bob, the second vertical block is $B_1, B_{m+1},...,B_{n-m+1}$ (colored red in the table) sent from Bob to Alice etc.

\begin{table}
\begin{center}
	\begin{tabular}{|c||l|l|l|l|l|l|l|}
		\hline block \# &
		\multicolumn{7}{ |c| }{transmissions} \\
		\hline\hline
		$1$ & \textcolor{blue}{$\Atrans_{1}$} & \textcolor{red}{$\Btrans_{1}$} & $\Atrans_{2}$ & $\Btrans_{2}$&$\ldots$&$\Atrans_{\blockLength}$ & $\Btrans_{\blockLength}$\\
		\hline
		$2$ & \textcolor{blue}{$\Atrans_{\blockLength+1}$} & \textcolor{red}{$\Btrans_{\blockLength+1}$} & $\Atrans_{\blockLength+2}$ & $\Btrans_{\blockLength+2}$&$\ldots$&$\Atrans_{2\blockLength}$ & $\Btrans_{2\blockLength}$\\
		\hline
		$\vdots$&$\vdots$        &               &               &              &      &             & $\vdots$\\
		\hline
		$n/m$ & \textcolor{blue}{$\Atrans_{\protLength-\blockLength+1}$} & \textcolor{red}{$\Btrans_{\protLength-\blockLength+1}$}& $\Atrans_{\protLength-\blockLength+2}$ & $\Btrans_{\protLength-\blockLength+2}$&$\ldots$&$\Atrans_{\protLength}$ & $\Btrans_{\protLength}$\\
		\hline\hline
		vertical block \#
		 &\textcolor{blue}{$1$}&\textcolor{red}{$2$}&$3$&$4$&$\cdots$&$2\blockLength-1$&$2\blockLength$\\ \hline
	\end{tabular}
\end{center}
\caption{Block-wise interaction\label{table:blocktrans}}
\end{table}

Note that now, Alice and Bob communicate in vertical blocks of length $n/m$, but the transmissions within every block are independent and hence non-interactive. Therefore a block code can be used for every block.
Also note, that the interaction is being carried out between consecutive blocks, which explains the term  \textit{block-wise interactive}. In addition, if the length of the vertical block code $\nblocks$, is \textit{large enough} and the number of vertical blocks $2\blockLength$, is \textit{not too large}, the probability of error in one (or more) block codes can go to zero with $\protLength$. This notion is made rigorous in the following basic lemma.
\begin{lemma}\label{lemma:blockwise} Suppose $l(n)$ independent blocks of $b(n)$ bits, are to be conveyed over BSC($\BSCprob$) at rate $\Rb<\CBSC$ and $n\to \infty$. Then, if $b(n) =\Omega(\log(l(n)))$, the probability of error in the decoding of one or more blocks is $o(1)$. 
\end{lemma}
\begin{proof}
	The proof is by straightforward implementation of Gallager's random coding error exponent and the union bound. Due to  \cite{GallagerIT}[Theorem 5.6.4], the probability of decoding error in a single block is upper bounded by:
	\begin{align}
	P(\text{block error})\leq \exp\left(-\frac{b(n)}{\Rb}E_r(\Rb)\right)
	\end{align}
	where $E_r(\Rb)$ (the error exponent) is strictly positive and independent of $n$ for any $0\leq \Rb<\CBSC$ and ${b(n)}/{\Rb}$ is the length of the block code. Now, having $l(n)$ independent such blocks, the probability of error in one or more blocks can be upper bounded using the union bound:
	\begin{align}
	P(\text{error in any block})&\leq l(n) \exp\left(-\frac{b(n)}{\Rb}E_r(\Rb)\right)\\
	&=\exp\left(-b(n)\frac{E_r(\Rb)}{\Rb}+\ln l(n)\right)\\
	&\stackrel{(a)}{=}\exp(-\Omega(1))\\
	&=o(1)
	\end{align}
	where $(a)$ is by the assumption that $b(n) =\Omega(\log(l(n)))$.
\end{proof}

Now assume that a vertical block code of rate $R_b$ is used, that $l(n)=2m$, and that $b(n) l(n)=2n$. The total number of transmissions is therefore $l(n)\frac{b(n)}{R_b}=\frac{2n}{R_b}$, so by \eqref{eq:ratedef}, the total rate of the scheme is $R=\Rb$. If $l(n)$ and $b(n)$ satisfy the condition in Lemma~\ref{lemma:blockwise}, then reliable simulation in any rate below Shannon's capacity in possible.

 We now provide two schemes that achieve the Shannon capacity using vertical coding. The schemes differ in the way the separation into blocks is implemented.
 
\section{First Capacity Achieving Scheme : a random partition into blocks\label{sec:randomgame}}
In the first step of the scheme, Alice partitions her transmission functions into segments, possibly of non-uniform length, all starting with a stuck function. Then she conveys the partitioning to Bob. This is done as follows:
\begin{enumerate}
	\item Regard $A_1$ as a stuck function, and repeat sequentially: suppose the last block so far starts with a stuck function at time $j$. Look for the first stuck function at time $k\geq j+\sqrt{n}$. Then the next block will start at time $k$. The length of each block created in this process is at least $\sqrt{n}$. The number of blocks is $p\leq \sqrt{n}$. 
	\item Alice communicates to Bob the starting point of every block. This can be regarded as a binary sequence of length $n$ with at most $\sqrt{n}$ ones. It can be described to Bob using no more than $(\sqrt{n}+1)\log(n)$, where the total number of ones, and the location of every one in the sequence are described using $\log(n)$ bits. We note that while more efficient compression methods can be considered, the overhead inflicted by this method is negligible when normalizing by $n$ and taking $n$ to infinity.
\end{enumerate}
	
After performing these steps, Alice and Bob are coordinated with regard to the partition into blocks as depicted in Fig.~\ref{fig:transblocks}. 
The figure divides the original transcript into two parts: Part A contains the first$ \sqrt{n}$ bits in each block and Part B contains all the rest, which can be of varying lengths. Note that by definition, there are no stuck functions in Part B. We will use a different communication protocol for each part.
There are two cases which are determined by the value of $p$.
\begin{enumerate}
	\item If $p>n^{\frac{1}{4}}$ we use block-wise interactive coding for Part A and non-interactive simulation for Part B.  In Part A we have $l(n)=2\sqrt{n}$ blocks each having $n(n)=p>n^{\frac{1}{4}}$ information bits. So, the condition of Lemma~\ref{lemma:blockwise} is satisfied and Part A can be reliably simulated at any rate below Shannon's capacity.	
	
	As for the non-interactive simulation, it is assured by construction that there are no stuck functions in Part B. Therefore, Alice can describe the transmission functions of Part B using a single bit each. These bits can be appended to the last block Alice sends in the block-wise interaction phase, and due to the large block length can be reliably decoded. Then, Bob can simulate his transmissions of Part B off-line and send them to Alice.	
	
	\item If $p<n^{\frac{1}{4}}$ non-interactive simulation can be used for both Part A and Part B. For Part A, two bits per transmission function are required and in Part B only a single bit per transmission function is required. The number of bits Alice needs in order to represent her functions is thus:
	\begin{align}
	\underbrace{2\cdot\sqrt{n}p}_{\text{Part A}}+\underbrace{1\cdot(n-\sqrt{n}p)}_{\text{Part B}}=
	n\left(1+\frac{p}{\sqrt{n}}\right)\leq 	n\left(1+n^{-\frac{1}{4}}\right).
	\end{align}
This term approached $n$ as $n\to \infty$, which implies that asymptotically Alice needs $n(1+o(n))$ in order to represent her functions. Bob then needs $n$ bits to send his part of the transcript (after simulating it off-lines). Both Alice's and Bob's transmissions can be sent in large blocks using capacity achieving codes, thus 
\begin{align}
\lim_{n\to\infty}\tilde{n}=\frac{2n}{1-h(\eps)}
\end{align}
which by the definition in \eqref{eq:ratedef} implies that $R$ is Shannon's capacity for $\text{BSC}(\eps)$.

\end{enumerate}

\begin{figure}
\begin{center}
	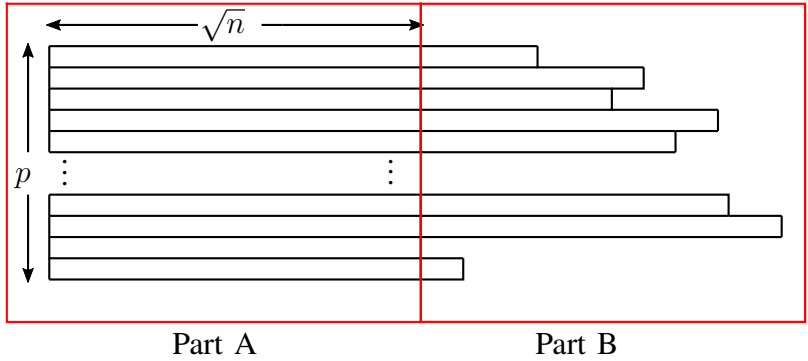
\end{center}
\caption{Division of the transmission functions into blocks \label{fig:transblocks}}
\end{figure}

\section{Second capacity achieving scheme: regular partition into blocks and prediction of the first transmission \label{sec:regularscheme}}
In this scheme, the blocks are constructed to be all equal length $m$ exactly as in Table~\ref{table:blocktrans}. 
We now give an algorithm that enables Alice to calculate her first transmission in every block efficiently, and without having to simulate the protocol. 
More explicitly, we give an algorithm, that for every block $i\in\{1,...,n/m\}$, comprising the transmissions $A_{(i-1)m+1},B_{(i-1)m+1},...,A_{i\times m},B_{i\times m}$, enables Alice to calculate Bob's last transmission 
$B_{i\times m}$, provided that $B_{(i-1)\times m}$ is given, and with the cost of exchanging $O(\log m)$ between her and Bob (we set $B_0=0$ without loss of generality). Once Alice knows $B_{(i-1)\times m}$, she can calculate 
her following transmission $A_{(i-1)\times m+1}$ using her transmission function.

We note that this algorithm can be regarded as a special pointer jumping game. But in contrast to the classic setting 
\cite{nisan1993rounds}, in which the alphabet size of the players goes to infinity, the number of rounds is fixed and the mapping functions are fixed, in our case the size of the alphabet is two, the number of rounds goes to infinity and the functions change over time.

The algorithm is implemented in the following steps for every block $i$:
\begin{enumerate}
	\item Bob conveys to Alice $s_B$, the index of the latest stuck function in the block. Namely, if there exists at least one value ${1\leq j\leq m}$, for which either $B_{(i-1)m+j}=0$ or $B_{(i-1)m+j}=1$, then:
	\begin{align}
	s_B = \max_{1\leq j\leq m} \text{s.t}\quad \{B_{(i-1)m+j}=0 \quad \text{or}\quad  B_{(i-1)m+j}=1\},
	\end{align}
	otherwise, $s_B=0$. In addition, Bob sends $v_B$, which is the value of his last stuck function, defined as $v_B=B_{s_B}$ if $s_B>0$ and $0$ otherwise.	
	All in all, $\lceil\log(m+1)\rceil$ are required for the description of $s_B$ and a single bit is required for $v_B$. So $\lceil\log(m+1)\rceil+1$ bits should be transmitted from Bob to Alice at this step.
	\item Alice conveys to Bob $s_A$, the index of the latest stuck function in her block. Namely, if there exists at least one value ${1\leq j\leq m}$, for which either $A_{(i-1)m+j}=0$ or $A_{(i-1)m+j}=1$, then:
	\begin{align}
	s_A = \max_{1\leq j\leq m} \text{s.t}\quad \{A_{(i-1)m+j}=0 \quad \text{or}\quad  A_{(i-1)m+j}=1\}.
	\end{align}
	if no such value exists, then $s_A=0$. This step requires the transmission of $\lceil\log(m+1)\rceil$ from Alice to Bob.
	\item Having the knowledge of $s_A$ and $s_B$ Alice and Bob can each calculate $s=\max(s_A,s_B)$, the index of the latest mutual stuck function.   
\end{enumerate}
The next steps rely on the fact that after $s$, there are no more stuck functions at either party. So for $s\leq j \leq m$ Bob's transmission functions are $B_{(i-1)m+j}=A_{(i-1)m+j}+b_{(i-1)m+j}$ and Alice's transmission functions are
$A_{(i-1)m+j}=B_{(i-1)m+j-1}+a_{(i-1)m+j}$ where $a_{(i-1)m+j},b_{(i-1)m+j}\in\{0,1\}$.

The next steps are divided into two cases, depending if $s=s_A$ or $s=s_B$. If $s=s_A$ (the latest mutual stuck function is at Alice's) and $s>0$ then do:
\begin{enumerate}
	\item Bob sends Alice the \textit{parity} of his transmission from $(i-1)m+s$ to ${im}$. Namely:
	\begin{align}
	p_B = \sum_{j={(i-1)m+s}}^{im} b_j
	\end{align}
	where $\sum$ and $+$ are both over $\mathbb{GF}(2)$.
	This process requires conveying a single bit from Bob to Alice.
	\item Alice can now calculate $B_m$ as follows:
	\begin{align}
	B_m = A_{s}+p_B+\sum_{j={(i-1)m+s+1}}^{im} a_j
	\end{align}	
\end{enumerate}
In the complementary case ($s=s_B$, including $s=0$) then
\begin{enumerate}
	\item Bob sends Alice the \textit{parity} of his transmission from $(i-1)m+s+1$ to ${im}$. Namely:
	\begin{align}
	p_B = \sum_{j={(i-1)m+s+1}}^{im} b_j
	\end{align}
	\item In the case where $s>0$ Alice can readily calculate $B_m$ as follows:
	\begin{align}\label{eq:Bm1}
	B_m = v_B+p_B+\sum_{j={(i-1)m+s+1}}^{im}a_j
	\end{align}
	In the case where $s=0$, $B_m$ corresponds to the following equation:	
	\begin{align}\label{eq:Bm2}
	B_m = B_{(i-1)m}+p_B+\sum_{j={(i-1)m+1}}^{im}a_j
	\end{align}
	We note that except the case where $i=1$ (for which we set $B_0=0$) Alice cannot calculate this value without knowing $B_{(i-1)m}$, and the explicit calculation will take place later.
\end{enumerate}

We note that in every Bob and Alice exchanged $2\lceil \log (m+1) \rceil+1$ bits altogether. The total number of bits is therefore:
\begin{align}
N=\frac{n}{m}(2\lceil \log (m+1) \rceil+1)
\end{align}
We can now set $m=\sqrt{n}$ and observe that
\begin{enumerate}
\item $N=o(n)$
\item Since by construction, the transmissions of every block are independent, they can be done in parallel. I.e. can implement the calculation in three rounds: Bob to Alice ($s_B$ and $v_B$ for all blocks), Alice to Bob ($s_A$ for all blocks), Bob to Alice ($p_B$ for all blocks). Each of there rounds contains $O(n/m\log m)$ bits, so it can be transmitted reliably using a block code.
\end{enumerate}

Finally, Alice needs to calculate $B_{im}$ for $1\leq i n/m-1$. She can do that inductively from $i=1$ to $i=n/m-1$ using \eqref{eq:Bm1} or \eqref{eq:Bm2} (according to the value of $s$ in the block). Once Alice knows $B_{(i-1)m}$ for all $i=1,...,n/m$, she can calculate her first transmission in every block (i.e. $A_{(i-1)m+1}$ for every $i=1,...,n/m$). Therefore, the protocol can be simulated using block-wise interaction with $l(n)=2\sqrt{n}$ blocks of length $b(n)=\sqrt{n}$. So, the condition in Lemma~\ref{lemma:blockwise} is satisfied and reliable simulation in any rate below Shannon's capacity in possible.

\section{Extension to higher order models}
The binary first order Markovian can be extended in two ways: either considering a binary model with a larger memory, or by leaving the Markovian model of the same order and extending the alphabet size. While both extensions are essentially similar, we prefer to discuss the second one. 

We note that our coding schemes presented in Sections~\ref{sec:randomgame} and \ref{sec:regularscheme} do not extend naturally to all possible protocols of a higher order, they are still applicable for \textit{most} protocols of a finite higher order.

To explain this point, assume that the transmission functions are drawn uniformly i.i.d over the space of all transmission functions. We modify the algorithms as follows:
\begin{enumerate}
\item For the scheme in Section~\ref{sec:randomgame}, always look for the next stuck function within an \textit{opportunity window} of size $o(\sqrt{n})$ (for example, set the length of the opportunity window to be $n^{1/4}$). If the procedure does not succeed, declare failure.
\item For the scheme in Section~\ref{sec:regularscheme}, look for the last stuck function in a block only within the last $o(\sqrt{n})$ (for example, $n^{1/4}$) time indexes of the block. Then, Bob describes all his functions after this time point, until the end of the block. The description therefore requires
\begin{align}
N=\frac{n}{\sqrt{n}}\left(2\lceil \log (\sqrt{n}+1) \rceil+Kn^{1/4}\right)
\end{align}
bits. $K$ is the number of bits required to describe a single function, and is constant, therefore $N=o(n)$ and the number of bit required for the calculation of the initial transmission in every block is negligible.
If a stuck function was not found inside the opportunity window (of size $o(\sqrt{n})$) at least for one block, failure is declared.
\end{enumerate}

It is not difficult to show, that if the transmission functions are drawn uniformly and i.i.d, then the procedure described will be be successful with high probability (or, for almost all protocols).

\bibliographystyle{IEEEbib}
\bibliography{bibtex_references}

\end{document}